\documentclass{article}

\usepackage{latexsym,amssymb,amsmath,amsthm}
\usepackage{epsf,graphicx}
\usepackage{algorithm}
\usepackage[noend]{algorithmic}
\renewcommand{\algorithmicrequire}{\textbf{Input:}}
\renewcommand{\algorithmicensure}{\textbf{Output:}}

\title{Performance of a greedy algorithm for edge covering by cliques in interval graphs}

\author{
Gabrio Caimi, Martin Fuchsberger \\
Institute for Operations Research, ETH Zurich \\
CH-8092 Z\"urich
\and
Holger Flier, Marc Nunkesser \\
Institute of Theoretical Computer Science, ETH Zurich \\
CH-8092 Z\"urich
}

\newtheorem{theorem}{Theorem}[section]

\newtheorem{definition}[theorem]{Definition}

\newtheorem{corollary}[theorem]{Corollary}
\newtheorem{property}[theorem]{Property}

\begin{document}

\maketitle

%\tableofcontents

\begin{abstract}
%We address the problem of conflict-free train scheduling in main station areas.
In this paper a greedy algorithm to detect conflict cliques in interval graphs and circular-arc graphs is analysed.
In a graph, a stable set requires that at most one vertex is chosen for each edge. It is equivalent to requiring that at most one vertex for each maximal clique is chosen.
We show that this algorithm finds all maximal cliques for interval graphs, i.e.~it can compute the convex hull of the stable set polytope.
In case of circular-arc graphs, the algorithm is not able to detect all maximal cliques, yet remaining correct.
This problem occurs in the context of railway scheduling.
A train requests the allocation of a railway infrastructure resource for a specific time interval. As one is looking for conflict-free train schedules, the used resource allocation intervals in a schedule must not overlap. 
The conflict-free choices of used intervals for each resource correspond to stable sets in the interval graph associated to the allocation time intervals.
\end{abstract}

\section{Introduction} \label{introduction}

\subsection{Motivation} \label{motivation}
One important aspect of railway optimization is scheduling trains in a main station area \cite{HuismanKroon_Survey_ORinRailways05,Survey06_Caprara}. 
In such an area the track topology is complex, and many different routes to travel between the area entrances (called portals) and platforms exist. Here, an appropriate assignment of exactly one of these routes to each train is crucial.
The choice of these routes has to be conflict-free, i.e.~each infrastructure resource can be occupied by at most one train at the same time. 
Hence, the allocation time intervals at each resource belonging to the chosen train routes can not overlap (\cite{Pachl}).
These restrictions over the concurrent allocation time intervals at each resource transforms into finding stable sets in corresponding interval graphs.
The stable set problem itself can be modeled by formulating an integer linear program (ILP) and solving it using a commercial solver, once an objective has been formulated. 
Since a good ILP formulation has a strong relaxation which speeds up the solution process, we look for an efficient ILP description of the stable set problem in an interval graph. 

\subsection{The basic idea}

Let consider the track topology, consisting in a set of resources, as described in \cite{montigel92:DoubleVertexGraph}. Each route/time assignment for a train allocates the used resources over certain time intervals, depending on signal positions, train dynamics, track topology, and additional safety regulations.
A feasible train schedule assigns a route and times to each train such that no resource is allocated at the same time, thereby guaranteeing a conflict-free schedule. 

Several approaches exist in the literature for modeling conflicting allocations over resources.
The conflict modeling in the conflict graph approach of \cite{kroon:trainroutingalg}
is simple: each available route/time assignment of a train corresponds to a node in the conflict graph. 
Each time two nodes of distinct trains would allocate a resource at the same time, a conflict edge between these nodes is introduced into the graph. Additionally, the fact that a train is only allowed to have one route/time assignment is modeled by interconnecting all nodes of the same train by edges thereby forming a clique. 
A solution of the train routing problem in this model is a stable set with cardinality equal to the number of trains $n = |T|$, where each chosen vertex assigns a route/time to the corresponding train. As it is a stable set of cardinality $n$, this assignment is guaranteed to be conflict-free and each train gets exactly one route/time assigned.
We can illustrate this conflict modeling concept by drawing all the allocation time intervals of a single resource during a period of time (see later Figure \ref{fig:allocationschema}).
This leads to the following ILP formulation for all the stable sets (feasible solutions) of the conflict graph $G$, where the variable $x_{ij}$ correspond to the node $ij \in G$, meaning the route/time assignment $j$ for train $i$: 

% \begin{figure*}[thb]
%   \centering
%   \includegraphics[scale=0.4]{pictures/conflictmodel1.png}
%   \caption{Allocation of a resource by several train routings}
%   \label{fig:TCGConflictModel}
% \end{figure*}

\begin{eqnarray}
\sum_{j=1}^{m(i)}x_{ij} = 1 & \textrm{for all } i=1,\ldots, n \label{traincliques} \\
x_{ik} + x_{jl} \le 1& \textrm{for all } r_{i_k}\nleftrightarrow r_{j_l} \label{pair_conflicts}\\
x_{ij} \in\{0,1\} \label{integrality}
\end{eqnarray}
where $r_{i_k}\nleftrightarrow r_{j_l}$ means that assignment $k$ of train $i$ is in conflict with assignment $l$ of train $j$. 

Let $$STAB1(G):=\{x | x \geq 0, x \textrm{ satisfies } (\ref{traincliques}), (\ref{pair_conflicts}) \}$$ be the polytope of the non-negative vectors $x$ satisfying (\ref{traincliques}) and (\ref{pair_conflicts}).
It is easy to see that integral solutions fulfilling (\ref{traincliques}), (\ref{pair_conflicts}), and (\ref{integrality}) are exactly the incidence vectors of stable sets of nodes of $G$.

Let denote with $$STAB(G):= \textrm{conv}\{x^S \in \{0,1\}^V | S \subseteq V \textrm{ is stable set} \}$$ the convex hull of the incidence vectors of all stable sets of nodes of G. 
Clearly, $STAB(G) \subseteq STAB1(G)$ and the integer points contained in both sets are the same.
The two polytopes are equal if and only if the graph $G$ is bipartite and has no isolated nodes \cite{GroetschelCombopt}. This is in general not the case when looking at train scheduling problems. Hence, we have to improve on the $STAB1(G)$ formulation.

\subsection{Gathering conflicts}

In the following we restrict our view to a single resource of the track topology, consisting in a set of topology elements \cite{montigel92:DoubleVertexGraph}.
Instead of looking at pairs of overlapping allocation time intervals like in the conflict graph approach, \cite{Thesis_Fuchsberger07,RTCG_paper07} introduces a more sophisticated attempt which gather them into groups of conflicting intervals.
All possible assignments using the resource at the same time are grouped in a conflict clique, and at most one of these can use the resource (see later Figure \ref{fig:resourceconflicts}).

% \begin{figure*}[th]
%   \centering
%   \includegraphics[scale=0.4]{pictures/conflictmodel2.png}
%   \caption{Improved conflict modeling grouping conflicts into cliques.}
%   \label{fig:RTCGConflictModel}
% \end{figure*}

Let the graph $G=(V,E)$ denote the induced interval graph from the allocation time intervals for this resource, which is built as follows: for each allocation time interval a vertex $v \in V$ is created and two vertices are connected with an edge $e \in E$ if the two intervals intersect.
On the other hand, in periodic scheduling we deal with circular time axes and the induced intersection graph $G$ is called \emph{circular-arc graph}.
The set of allocation time intervals is called \emph{arc model}.

A feasible solution of the train scheduling problem should be conflict-free in each resource. It means that the set of chosen assignments represent a stable set in the interval graph of each resource. 
The conflict graph formulation describes the set of all stable sets in the graph $G$ using the incidence vectors and (integer) linear constraints avoiding the choice of two adjacent nodes.
Ideally, one would be able to describe efficiently the convex hull of the incidence vectors of all stable sets in $G$ as a polytope. Since the polytope is naturally integer it induces an efficiently solvable ILP formulation. 
Chv\'atal \cite{Chvatal75} proved that one can describe efficiently this convex hull. 
This description imposes that the amount of stable set vertices over each maximal clique of the interval graph should be at most one.

Section~\ref{greedy} explains the greedy algorithm to detect the conflict cliques and in Section~\ref{proof} 
we prove that this algorithm finds all maximal cliques for interval graphs (i.e.~describes the convex hull).
\section{A greedy algorithm for grouping conflicts into cliques} \label{greedy}
In this section we explain the algorithm introduced by \cite{Thesis_Fuchsberger07,RTCG_paper07} for detecting conflicts and grouping them into cliques.

\subsection{Allocation Schema}
For an easier understanding of the algorithm we first introduce a graphical representation of resource allocations called Allocation Schema (AS).
We create for each resource an AS, which is a representation of the resource allocation intervals by the potential assignments over time. We chart the allocations of the different assignments by a horizontal line with the start of the allocation as the offset and the allocation duration as the length of the line. An example of AS is given in Figure~\ref{fig:allocationschema}. 
% FIGURE: Allocation Schema
\begin{figure}[th]
\centering
\includegraphics[width=0.6 \textwidth]{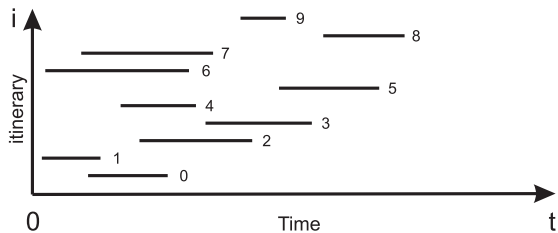}
\caption[Allocation Schema (AS)]{In the Allocation Schema (AS) the allocations of a resource by train assignments are represented by horizontal lines.}
\label{fig:allocationschema}
\end{figure}
% END FIGURE

\subsection{Algorithm for finding conflict cliques}
To find all conflict cliques we suggest an algorithm which operates on the AS and not on the intersection graph directly. 
Algorithm~\ref{alg:conflictcliques} basically consists of two major steps:
\begin{enumerate}
\item Sort the start and end times of the allocation time intervals according to time.
\item Walk through the times, keep track of the currently open intervals in a list and construct the set of conflict cliques among the currently open intervals at the end times.
\end{enumerate}
Note that we do not have to form a clique at the end time of each interval, but only at end times where a new interval has been opened since the last iteration.
Additionally, we avoid adding a clique if it just consists of one interval since it would not be a necessary constraint.
A result of Algorithm~\ref{alg:conflictcliques} is illustrated in Figure~\ref{fig:resourceconflicts}.
% ALGORITHM: Conflict Cliques
\algsetup{indent=2em}
\begin{algorithm}[th]
\caption{Conflict Cliques}\label{alg:conflictcliques}
\begin{algorithmic}[1]
\REQUIRE Set of all allocation time intervals $\bigl\{[s(i),f(i)]\bigr\}_{i=1,\ldots,n}$ of a resource $r$
\ENSURE Set $C^r$ of conflict cliques
\STATE Create a list $L$ of $2n$ tupels (\textit{time} $\in \mathbb{R}^{+}$, \textit{identifier} $\in \{1,\ldots,n\}$, boolean \textit{is\_endtime}) 
\STATE Sort $L$ according to the first key \textit{time}, for same time starttimes before endtimes 
\STATE Initialize list of currently open intervals $O := \emptyset$
\STATE Initialize $\textit{new\_starttime} := FALSE$
\FOR {$i = 1$ to $2n$}
	\IF {$\textit{is\_endtime}^i = FALSE$}
		\STATE $O := O \cup \{\textit{identifier}^i\}$
		\STATE $\textit{new\_starttime} := TRUE$ 
	\ELSE
		\IF {$(\textit{new\_starttime} = TRUE)\, \&\& \,(|O| > 1)$}
			\STATE  $C^r := C^r \cup O$
			\STATE $\textit{new\_starttime} := FALSE$
		\ENDIF
		\STATE $O = O \setminus \{\textit{identifier}^i\}$
	\ENDIF
\ENDFOR
\RETURN $C^r$
\end{algorithmic}
\end{algorithm}
%END ALGORITHM
%FIGURE: Conflict Sets
\begin{figure}[th]
\centering
\includegraphics[width=0.6 \textwidth]{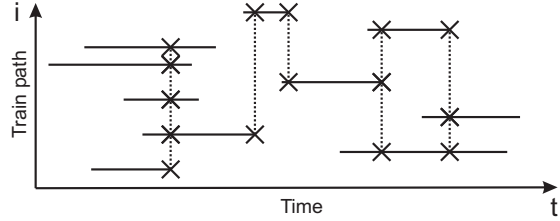}
\caption[Concurrent resource occupation]{Several train itineraries allocate a resource during different time intervals (horizontal bars). The dashed vertical bars illustrate conflict sets among the allocation intervals formed by Algorithm~\ref{alg:conflictcliques}.}
\label{fig:resourceconflicts}
\end{figure}
% END FIGURE

\subsection{Complexity and remarks}
Algorithm~\ref{alg:conflictcliques} runs in $O(n \log n)$ where $n$ is the number of allocation intervals. In a first step we have to sort the allocation intervals according to time 
which takes $O(n \log n)$ time \cite{Ottmann_AlgoDatenstrukturen}. We then loop $2n$ times where the steps inside  the loop take $O(1)$ time which results in a running time complexity of $O(n)$ for the loop and $O(n \log n)$ for the whole algorithm.

For non-periodic railway scheduling problems the described algorithm is optimal, i.e.~it finds all maximal cliques in the corresponding interval graph and thereby also the convex hull of the stable set incidence vectors. The proof follows in Section~\ref{proof}. 
In periodic railway scheduling problems, the allocations of a resource reoccur after a fixed period length and hence allocation intervals may overlap the period length. 
To deal with these overlaps of the period length, we can use the modulo function to draw overlapping intervals by two horizontal lines. The first line starts at time $0$ and ends at the original end time of the allocation modulo the period length. The second line starts at the original start time and ends at the period length. 

There are two drawbacks in this periodic setting: the first is the unknown quality of the polytope described by the maximal clique.
For instance, in the chordless odd cycles $K_{2k+1}$ for $k\geq 2$ the non-integer solution where all vertices have assigned the value $\frac{1}{2}$ is an extremal point of the polytope described by the maximal cliques, making it in general not the convex hull for the incidence vectors of a stable set.
It is unclear how one can find efficiently the convex hull of the stable set incidence vectors in the periodic setting and how good the maximal clique description approximates this convex hull. 

The second problem is, that the current greedy algorithm fails to detect all maximal conflict cliques. For instance, in the example of Figure \ref{fig:cmproblem} the described method fails to detect a conflict clique:
Typically, in practice the length of a period is one hour and the allocation time intervals are not longer than 3 minutes. Hence these special cases where a conflict clique is not detected by the algorithm most likely will not occur. 

\begin{figure*}[th]
	\centering
	\includegraphics[scale=0.4]{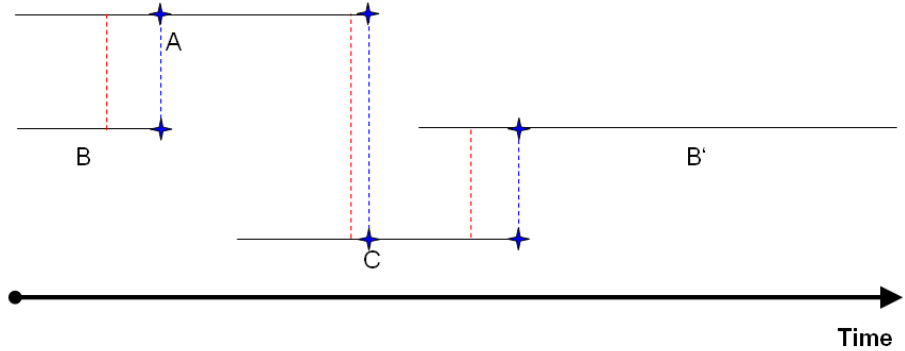}
	\caption[Finding Cliques: Special case]{Special case, where the three clique consisting of the intervals A,B, and C is not detected}
	\label{fig:cmproblem}
\end{figure*}

\section{Optimality proof for the non-periodic case} \label{proof}

In this section we present a proof that the greedy algorithm is optimal in case of non-periodic scheduling.

\begin{definition}
For a time interval $i$ of a resource $r$, denote its starting and finishing time by $s(i)$  and $f(i)$, respectively. An interval $i$ is called \emph{active at time} $t$ iff $s(i) \leq t \leq f(i)$.
\end{definition}

\begin{property}\label{prop:nonperiodic}
A subset $J_C$ of all allocation intervals $I^r$ of a resource $r$ forms a clique $C_J$ of the associated interval graph $G^r$, iff there exists a time segment $$\displaystyle T_C = \left[ s(J_C) := \max_{i \in J_C} \left\{ s(i) \right\},f(J_C) := \min_{i \in J_C} \left\{ f(i) \right\} \right] $$ during which all intervals $i \in J_C$ are active. Furthermore, for any $t < s(J_C)$ or $t > f(J_C)$, at least one of the intervals $i \in J_C$ is not active.
\end{property}

\begin{theorem} Algorithm~\ref{alg:conflictcliques} computes a minimum edge clique cover for the interval graph $G^r$ associated with the allocation intervals $I^r$. Each clique of the edge clique cover is maximal.
\end{theorem}

\begin{proof}
Clearly, $C^r$ computed by Algorithm~\ref{alg:conflictcliques} is an edge clique cover. Furthermore, every clique $C \in C^r$ is maximal. Suppose some $C \in C^r$ was not. Then, there would exist an interval $a \not\in J_C$, conflicting with every interval $b \in J_C$. But because the clique $C$, which was added to $C^r$ in line 11 of Algorithm~\ref{alg:conflictcliques}, is equal to the set $O_C$ of all open (i.e., active) intervals at time $t = f(J_C)$, and $a \not\in O_C$, it must hold that either $s(a) > f(J_C)$ or $f(a) < s(J_C)$, where $s(J_C)=s(\text{identifier}^i)$, the starting time of interval $\text{identifier}^i$ defined in line 7 of the algorithm. By Property~\ref{prop:nonperiodic}, both cases contradict the assumption that $a$ conflicts with every interval $b \in J_C$. \footnote{In the algorithm, a suitable tie break should be chosen, so that for intervals $i$, $j$ with $f(i) = s(j)$, the starting time of $j$ is considered before the finishing time of $i$.}

Finally, to show that $\ell := |C^r|$ is minimum, it suffices to identify a subset $S$ of the conflicts that requires at least $\ell$ cliques to be covered.  This subset can be constructed inductively as follows. Denote by $C_i$, $1 \leq i \leq \ell$ the cliques of $C^r$ in the order in which they are found by the algorithm. For every clique $C_i$, denote by $s_i$, $f_i$  intervals for which $s(C_i)=s(s_i)$ and $f(C_i)=f(f_i)$, respectively. Let $c_i$ be the conflict between $s_i$ and $f_i$, if $s_i \neq f_i$, or an arbitrary but fixed conflict between $s_i=f_i$ and another interval $a \in C_i$.  In either case, for any $i < j$, conflicts $c_i$ and $c_j$ cannot be covered by the same clique, since $f(f_i) < s(s_j)$, and hence, intervals $f_i$ and $s_j$ are not in conflict. Therefore, the set $S=\left\{c_i| 1 \leq i \leq \ell \right\}$ is a subset of conflicts that are pairwise not coverable by the same clique. Hence, any edge clique cover must consist of at least $|S|=\ell$ cliques.
%\hfill $\Box$
\end{proof}

\begin{corollary}
Algorithm~\ref{alg:conflictcliques} finds all maximal cliques in the interval graph $G^r$.
\end{corollary}

\begin{proof}
Suppose there exists a maximal clique $M \not\in C^r$. Clearly, $M$ can neither be a superset nor a subset of any $C \in C^r$. Hence, $M$ can only contain subsets of intervals of two consecutive cliques $C_i, C_{i+1} \in C^r$, i.e., $M=C'_i \cup C'_{i+1}$, with nonempty subsets $C'_i \subset C_i$ and $C'_{i+1} \subset C_{i+1}$. Now, $M \not\subset C_i$ can only hold if there is an interval $a$ in $C'_{i+1}$ that is not in $C'_i$. $M$ is only a clique, if $a$ is in conflict with every $b \in C'_i$. But if this is the case, then $M \subset C_{i+1}$, since all intervals that are open at $s(a)$ are, by Algorithm~\ref{alg:conflictcliques}, in $C_{i+1}$.
\end{proof}

\bibliographystyle{plain}
\bibliography{/ifor/u/caimig/Railroad/Bibliography/literature}

\end{document}